\documentclass{article}

\usepackage{amsmath,amsfonts,amsthm,amssymb}
\usepackage{algorithm}
\usepackage{array}
\usepackage[caption=false,font=normalsize,labelfont=sf,textfont=sf]{subfig}
\usepackage{textcomp}
\usepackage{stfloats}
\usepackage{url}
\usepackage{verbatim}
\usepackage{graphicx}
\usepackage[left=1in,right=1.2in,top=1in,bottom=1in,letterpaper]{geometry}
\usepackage{color}

\newcommand{\R}{\mathbb{R}}
\newcommand{\norm}[1]{\lVert#1\rVert}
\newcommand{\st}{\mbox{s.t.}}
\newtheorem{theorem}{Theorem}
\newtheorem{definition}{Definition}

\usepackage{cite}
\usepackage{amsmath,amssymb,amsfonts}
\usepackage{algorithmic}
\usepackage{graphicx}
\usepackage{textcomp}
\usepackage{xcolor}

\newcommand{\va}{\mathbf{a}}

\newcommand{\ve}{\mathbf{e}}

\newcommand{\vu}{\mathbf{u}}

\newcommand{\vw}{\mathbf{w}}
\newcommand{\vx}{\mathbf{x}}
\newcommand{\vy}{\mathbf{y}}
\newcommand{\vz}{\mathbf{z}}

\newcommand{\mA}{\mathbf{A}}
\newcommand{\cH}{\mathcal{H}}

\newcommand{\vA}{\mathbf{A}}

\DeclareMathOperator*{\argmin}{argmin}

\DeclareMathOperator*{\supp}{supp}
\usepackage{authblk}

\begin{document}

\title{Stochastic Natural Thresholding Algorithms
}

\author[1]{Rachel Grotheer}
\author[2]{Shuang Li}
\author[3]{Anna Ma}
\author[2]{Deanna Needell}
\author[4]{Jing Qin\thanks{Corresponding author: \texttt{jing.qin@uky.edu}}}
\affil[1]{Dept. of Mathematics, Wofford College, Spartanburg, USA}
\affil[2]{Dept. of Mathematics, University of California, Los Angeles, USA}
\affil[3]{Dept. of Mathematics, University of California, Irvine, USA}
\affil[4]{Dept. of Mathematics, University of Kentucky, Lexington, USA}

\date{}
\maketitle

\begin{abstract}
Sparse signal recovery is one of the most fundamental problems in various applications, including medical imaging and remote sensing. Many greedy algorithms based on the family of hard thresholding operators have been developed to solve the sparse signal recovery problem. More recently, Natural Thresholding (NT)
has been proposed with improved computational efficiency. This paper proposes and discusses convergence guarantees for stochastic natural thresholding algorithms by extending the NT from the deterministic version with linear measurements to the stochastic version with a general objective function. We also conduct various numerical experiments on linear and nonlinear measurements to demonstrate the performance of StoNT.
\end{abstract}

\section{Introduction}
In various fields, such as machine learning, computer vision, and signal processing, there is a widespread need to make inferences about data with a high number of dimensions, even when only limited measurements are available.
Effective algorithms for data inference from a limited number of measurements often rely on the observation that even though most real-world data exists in high-dimensional spaces, they often possess a low-dimensional complexity, such as sparsity.
Many signal recovery algorithms have been developed to exploit sparsity with promising effectiveness and efficiency for data inference and recovery.

In the sparse signal recovery, the underlying data $\vx \in \mathbb{R}^{n}$ is typically recovered by solving an optimization problem of the form
\begin{equation}
\min_\vx f(\vx) \quad\st\quad ||\vx||_0 \leq k,
\label{eq:genobj}
\end{equation}
where the objective function, $f(\vx)$, measures the model discrepancy and $k$ is a preassigned sparsity level of $\vx$.
For example, compressed sensing assumes the measurements are linearly related to the underlying sparse signal up to noise, which has the objective function
\begin{equation}
    f(\vx) = || \mA\vx - \vy||_2^2
\label{eq:linobj}
\end{equation}
where $\mA \in \mathbb{R}^{m\times n}$ is the sensing matrix, and $\vy := \mA\vx + \nu\in \mathbb{R}^{m}$ is the vector of measurements, with Gaussian noise $\nu$.

The optimization problem in~(\ref{eq:genobj}) can be solved using greedy iterative methods that employ thresholding operators.
Thresholding algorithms are particularly effective at solving these optimization problems due to their low computational complexity.
To enforce the sparsity, a thresholding operator is usually involved to either restrict the support of the estimated solution at each iteration with a fixed cardinality or approximate the support of the actual solution through iterations. For example, Iterative Hard Thresholding (IHT)~\cite{blumensath2009iterative} and its variants \cite{herrity2006sparse, blumensath2008iterative, foucart2011hard}, and Gradient Matching Pursuit (GradMP)~\cite{nguyenunified} which are based on the hard thesholding operator have shown the promising performance in many applications.
Several other types of thresholding operators exist, such as soft thresholding~\cite{donoho1995noising, bredies2008linear} and optimal $k$-thresholding (OT)\cite{zhao2020optimal}. More recently, natural thresholding~\cite{zhao2022natural} has been proposed to significantly reduce the computational cost of OT.

Specifically, to solve \eqref{eq:genobj}, the application of gradient descent and thresholding operator yields the IHT with the following iterative algorithm
\[
\vx^{(i+1)}=\cH_k(\vx^{(i)}-\lambda \nabla f(\vx^{(i)}))
\]
where $\cH_k$ is a hard thresholding that sets all but the largest $k$ components of a vector to zero, $\nabla f(\vx^{(i)})$ is the gradient of $f$ at $\vx^{(i)}$, and $\lambda>0$ is the step size. In the linear case \eqref{eq:linobj}, $\nabla f(\vx)=\mA^T(\mA\vx-\vy)$. However, the IHT type of algorithms easily cause numerical instability when the hard thresholding is independent of the objective function, especially in the linear case \cite{zhao2020optimal}. To address this issue, OT selects the $k$ components of a vector that achieves the least residual among all possible $k$-sparse selections.

To further enhance the performance of OT, the Natural Thresholding algorithm (NT) restricts the gradient of the regularized objective function of the OT given in~\cite{zhao2020optimal} to its $k$-smallest elements.
The regularized objective function is given by
\begin{equation}
g_\alpha(\vw) = ||\vy-\mA(\vu\otimes \vw)||_2^2 + \alpha \phi(\vw),
\end{equation}
where $\vu\in \R^n$ is a given vector, $\otimes$ is the Hadamard multiplication, $\vw$ is a binary vector, $\alpha$ is the regularization parameter and $\phi(\vw)$ is the regularization function that enforces the binary condition on $\vw$.
The Natural Thresholding Pursuit algorithm (NTP) is an extension of the NT algorithm that includes an orthogonal projection in the last step. Both algorithms are given in Algorithm~\ref{alg:NTP}, where a general objective function is used while only a linear objective function was presented in~\cite{zhao2022natural}.

\begin{algorithm}
\caption{Natural Thresholding (NT) and Natural Thresholding Pursuit (NTP)}\label{alg:NTP}
\begin{algorithmic}
\STATE{\textbf{Inputs:} $\vx^{(0)}$, sparsity level $k$, stepsize $\lambda$, tolerance $\varepsilon$, regularization parameter $\alpha>0$, maximum number of iterations $T$}
\FOR{$i=1,2,\ldots,T$}
\STATE
\vspace{-0.5cm}
\begin{align*}
&\vu^{(i)}=\vx^{(i)}-\lambda\nabla f(\vx^{(i)})\\
&\vw^-=\argmin_{\vw\in\{0,1\}^n}\norm{\vw-\vu^{(i)}}_2\\
&\nabla g_{\alpha}(\vw^-)=\nabla f(\vw^-\otimes \vu^{(i)})+\alpha\nabla \phi(\vw^-)\\
&\vw^+=\argmin_{\vw\in\{0,1\}^n,\,\ve^T\vw=k}\nabla g_\alpha(\vw^-)^T\vw\\
&S^{(i)}=\supp(\vw^+\otimes \vu^{(i)})
\\
&\vx^{i+1}=
\left\{
\begin{aligned}
&\vw^+\otimes \vu^{(i)} &&(NT)\\
&\argmin_{\supp(\vz)\subseteq S^{(i)}}f(\vz) &&(NTP)\\
\end{aligned}
\right.
\end{align*}
\ENDFOR
\end{algorithmic}
\end{algorithm}

When the data size is growing, stochastic versions of these thresholding algorithms, such as stochastic GradMP (StoGradMP) and stochastic IHT (StoIHT)~\cite{nguyen2017linear}, have the benefit of reduced computational complexity and running time.
Here the objective function $f(\vx)$ is assumed to be separable, that is,
$f(\vx) = \sum_i f_i(\vx).$
At each iteration, a small subset of indices $n_i$ are randomly chosen, and the gradient is computed only for the $f_i$ where $i \in n_i$.

In this paper, we propose two new algorithms--stochastic Natural Thresholding (StoNT) and Stochastic Natural Thresholding Pursuit (StoNTP).
These algorithms are the respective stochastic version of NT and NTP proposed by~\cite{zhao2022natural}.
The convergence of our algorithm is discussed when solving~(\ref{eq:genobj}) with the objective function given by~(\ref{eq:linobj}). A variety of numerical simulations have shown that StoNTP converges faster than the NTP algorithm with proper parameters.

\section{Stochastic Iterative Natural Thresholding}
\label{sec:algintro}
Before introducing our algorithms, we provide the necessary assumptions for the objective function.
First, we require that $f$ satisfy the restricted strong convexity (RSC) condition and that each of the $f_i$ satisfy the restricted strongly smooth condition.

\begin{definition}[RSS]
A function $f:\R^n\to \R$ is called restricted strongly smooth (RSS) with a constant $\rho_k^+>0$ if the following condition is satisfied
\[
\norm{\nabla f(\vx)-\nabla f(\vx')}_2\leq \rho_k^+\norm{\vx-\vx'}_2
\]
for any $\vx,\vx'\in\R^n$ with $|\supp(\vx')\cup\supp(\vx)|\leq k$.
\end{definition}

\begin{definition}[RSC]
A function $f:\R^n\to\R$ is called restricted strongly convexity (RSC) with a constant $\rho_k^->0$ if the following condition is satisfied:
\[
f(\vx')-f(\vx)-\langle \nabla f(\vx),\vx'-\vx\rangle\geq
\frac{\rho^-_k}2\norm{\vx'-\vx}_2^2
\]
for any $\vx',\vx\in\R^n$ with $|\supp(\vx')\cup\supp(\vx)|\leq k$.
\end{definition}

\subsection{Proposed Algorithms}
Given a function $f:\mathbb{R}^n\to\R$ which is differentiable and separable, i.e.,
\[
f(\vx)=\sum_{i=1}^n f_i(\vx),\quad n\in\mathbb{N},
\]
we consider the sparsity-constrained minimization problem
\begin{equation}\label{eqn:model}
\min_{\vx} f(\vx)\quad\st\quad
\norm{\vx}_0\leq k
\end{equation}
where $k\in\{1,2,\ldots,n\}$.  By letting $\vx=\vu\otimes \vw$, the sparsity constraint can be recast as
\[
\ve^T\vw= k,\quad \vw\in\{0,1\}^n,
\]
where $\ve=[1,1,\ldots,1]^T\in\R^n$.

We propose two new algorithms, the Stochastic Natural Thresholding (StoNT) algorithm, and the Stochastic Natural Thresholding Pursuit (StoNTP) algorithm, described in Algorithm~\ref{alg:stoNTP}.

\begin{algorithm}
\caption{Stochastic Natural Thresholding (StoNT) and Stochastic Natural Thresholding Pursuit (StoNTP)}\label{alg:stoNTP}
\begin{algorithmic}
\STATE{\textbf{Inputs:} $\vx^{(0)}$, sparsity level $k$, stepsize $\lambda$, probability $p(n_i)$, tolerance $\varepsilon$, regularization parameter $\alpha>0$, maximum number of iterations $T$}
\FOR{$i=1,2,\ldots,T$}
\STATE
\vspace{-0.5cm}
\begin{align*}
&\mbox{Randomly select an index or a batch of indices }n_i\\
&\mbox{ with a probability } p(n_i)\\
&\vu^{(i)}=\vx^{(i)}-\frac{\lambda}{np(n_i)}\nabla f_{n_i}(\vx^{(i)})\\
&\vw^-=\argmin_{\vw\in\{0,1\}^n}\norm{\vw-\vu^{(i)}}_2\\
&\nabla g_{\alpha}(\vw^-)=\nabla f(\vw^-\otimes \vu^{(i)})+\alpha\nabla \phi(\vw^-)\\
&\vw^+=\argmin_{\vw\in\{0,1\}^n,\,\ve^T\vw=k}\nabla g_\alpha(\vw^-)^T\vw\\
&S^{(i)}=\supp(\vw^+\otimes \vu^{(i)})
\\
&\vx^{i+1}=\left\{\begin{aligned}
&\vw^+\otimes \vu^{(i)}&&(StoNT)\\
&\argmin_{\supp(\vz)\subseteq S^{(i)}}f(\vz)&&(StoNTP)\\
\end{aligned}\right.
\end{align*}
\ENDFOR
\end{algorithmic}
\end{algorithm}

\section{Theoretical Guarantees}
\label{sec:theory}

In this section, we will focus on the linear measurement case for convergence analysis, which can be further extended to the nonlinear case. Consider $f(\vx) = \|\mA \vx - \vy \|_2^2$ where $ \vy = \mA\vx^* + \nu$ and $\|\vx^*\|_0 \leq k$ where $\mA \in \mathbb{R}^{m\times n}$ ($m \ll n$) satisfies the RIP Condition for $k$-sparse vectors with RIP constant $\delta_k$.

\begin{theorem} (Linear Convergence of StoIHT~\cite[Theorem 1] {nguyen2017linear}) Let $\vx_s$ be a feasible solution of
$$\min_{\vx} \frac{1}{m} \sum_{i=1}^m f_i(\vx) \quad\st\quad \| \vx\|_0 \leq k. $$
Suppose that $i \sim [m]$ with probability $p(i)$ and let
$$\vx^{t+1} = \mathcal{H}_k \left(\vx^t - \frac{\lambda}{mp(i)} \nabla f_{i}(\vx_t)\right).$$
If $\lambda < 2/\alpha_{3k}$ then:
\begin{equation}   \mathbb{E}\|\vx^{t+1} - \vx_S \|_2 \leq \kappa \|\vx^t - \vx_S \|_2 + \sigma_{\vx_S},
\end{equation}
where $\kappa$ and $\sigma_{\vx_S}$ are constants that depend on the RSS and RSC constant and $\alpha_k = \max_i \frac{\rho_k^+(i)}{mp(i)}$.
\label{thm:stoIHT}
\end{theorem}

\begin{theorem}
Assume the rows of $\vA$ have unit norms. Consider Algorithm \ref{alg:stoNTP} with batch size bs=1, and choose $\lambda < 2/\alpha_{3k}$ where $\alpha_{3k} = \max_i \frac{\rho_{3k}^+(i)}{mp(i)}$. Then
$$ \mathbb{E} \|\vx_S - \vx^{(p+1)} \|_2 \leq \kappa_\text{new} \| \vx_S - \vx^{(p)}  \|^2 + \sigma_\text{new},$$
where $\kappa_\text{new} = \sqrt{\frac{1 + \delta_{2k}}{1 - \delta_{2k}}} \kappa$ and $\sigma =  \frac{\sqrt{1 + \delta_{2k}}\sigma_{\vx_S} + 2\| \nu'\|_2}{\sqrt{1 - \delta_{2k}}}.$
\end{theorem}

\begin{proof}
Starting with Eq. (34) in ~\cite{zhao2022natural}, we have:
\begin{align*}
    &\mathbb{E} \|\vx_S - \vx^{(p+1)} \|_2 \\
    &\leq \sqrt{\frac{1 + \delta_{2k}}{1 - \delta_{2k}}} \mathbb{E}  \|\vx_S - \mathcal{H}_k(u^{(p)}) \|_2 +  \frac{2\| \nu'\|_2}{\sqrt{1-\delta_{2k}}} \\
    &\leq \sqrt{\frac{1 + \delta_{2k}}{1 - \delta_{2k}}} \kappa \| \vx_S - \vx^{(p)} \|_2 + \frac{\sqrt{1 + \delta_{2k}}\sigma_{\vx_S} + 2\| \nu'\|_2}{\sqrt{1 - \delta_{2k}}}.
\end{align*}
where in the first inequality, we are taking an expectation conditional on the first $p$ iterations of Algorithm~\ref{alg:stoNTP}, and in the second inequality, we use Theorem~\ref{thm:stoIHT}. Iterating the expectation obtains the desired result.
\end{proof}

\section{Numerical Experiments}
\label{sec:results}
Various experiments on linear and nonlinear measurements are conducted to evaluate the proposed performance. We adopt the following two comparison metrics: (1) relative error $\norm{\vx-\vx^*}_2/\norm{\vx^*}_2$ where $\vx$ is an approximation of the ground truth vector $\vx^*$; (2) success rate which is a percentage of successful cases with correctly identified support out of the total trials.
Numerical experiments were run on a 2015 Macbook Pro in MATLAB R2017b with 8 GB RAM and a 2.7 GHz Dual-Core Intel Core i5.

\subsection{Linear Measurements}
First, we illustrate the performance of StoNTP on the least squares problem, where the objective function is given as in~\eqref{eq:linobj}. We generate $\vx^\star\in\R^{800}$ as a normalized sparse Gaussian random vector with $10$ uniformly distributed nonzero entries. The sensing matrix $\mA\in\R^{100\times 800}$ is generated as a Gaussian random matrix with normalized columns. We then get the random measurements as $\vy = \mA \vx^\star$. We set the maximal number of iterations as 150 and the batch size as 10. The algorithm stops either when it achieves the maximal number of iterations or the loss function $\|\vy-\mA\vx\|_2 \leq 10^{-3}$. To see the best choice of the regularization parameter $\alpha$, we first fix the step size $\lambda = 2$. The value of the loss function and distance between the estimated $\vx$ and $\vx^\star$ evaluated at each iteration and versus the running time are illustrated in Fig.~\ref{fig:StoNTP_a}. It can be seen that the best choice is $\alpha = 1$. Next, we repeat the experiment by fixing $\alpha = 1$ and test on a variety of $\lambda$ values. As is shown in Fig.~\ref{fig:StoNTP_lamb}, the best step size is $\lambda = 2$.  We also compare our StoNTP algorithm with the NTP algorithm. It can be seen from Fig.~\ref{fig:NTPvsStoNTP} that the StoNTP algorithm significantly outperforms the NTP algorithm. In addition, we test the success rates for NTP and StoNTP for various parameters in Fig.~\ref{fig:tune_para}.

\begin{figure}
\centering
\includegraphics[width = 0.44\textwidth]{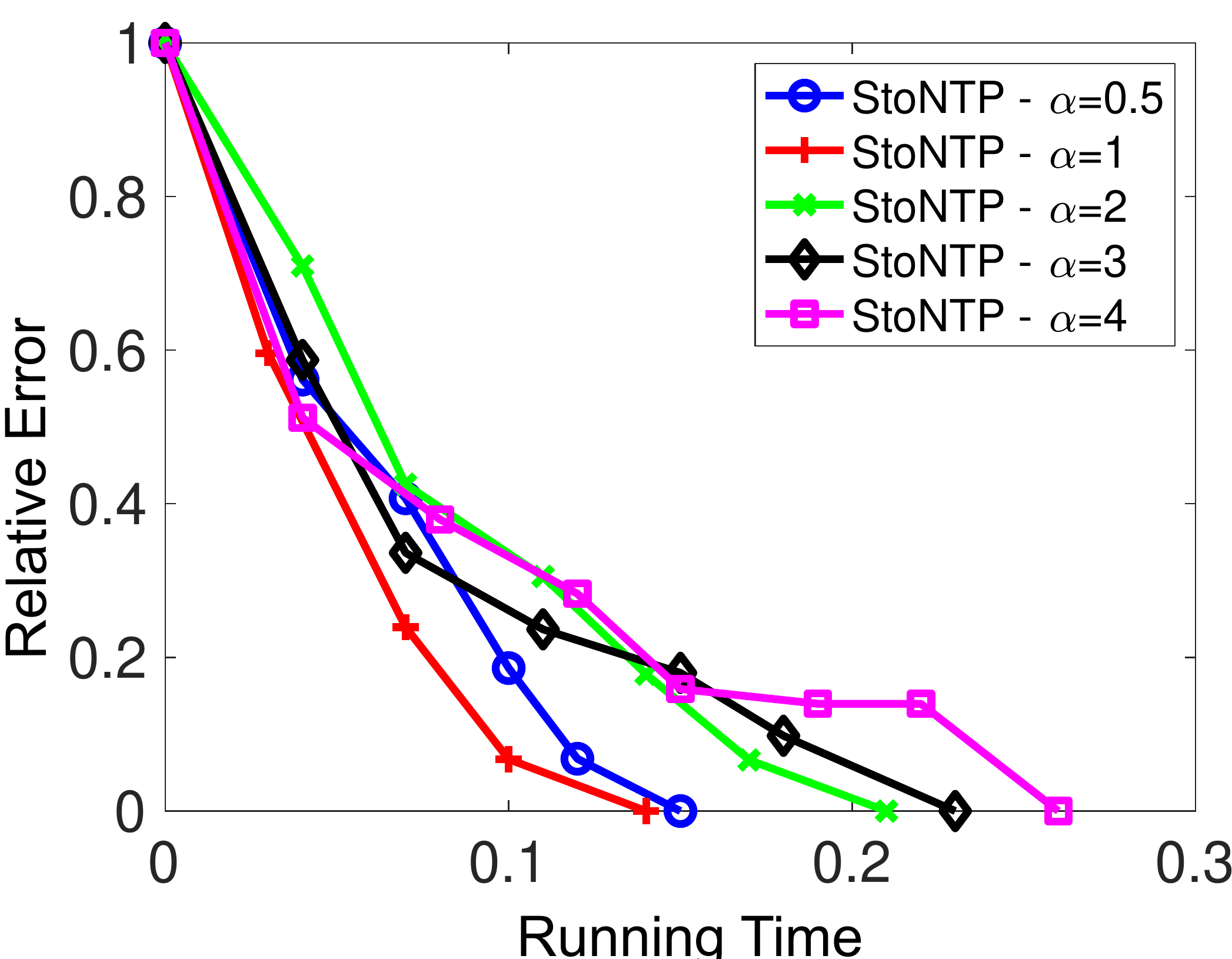}
\caption{Test StoNTP with various $\alpha$'s: $m=100$, $n=800$, $k=10$, $\lambda = 2$. Batch size for StoNTP is 10. The best choice is $\alpha = 1$.}
\label{fig:StoNTP_a}
\end{figure}

\begin{figure}
\centering
\includegraphics[width = 0.44\textwidth]{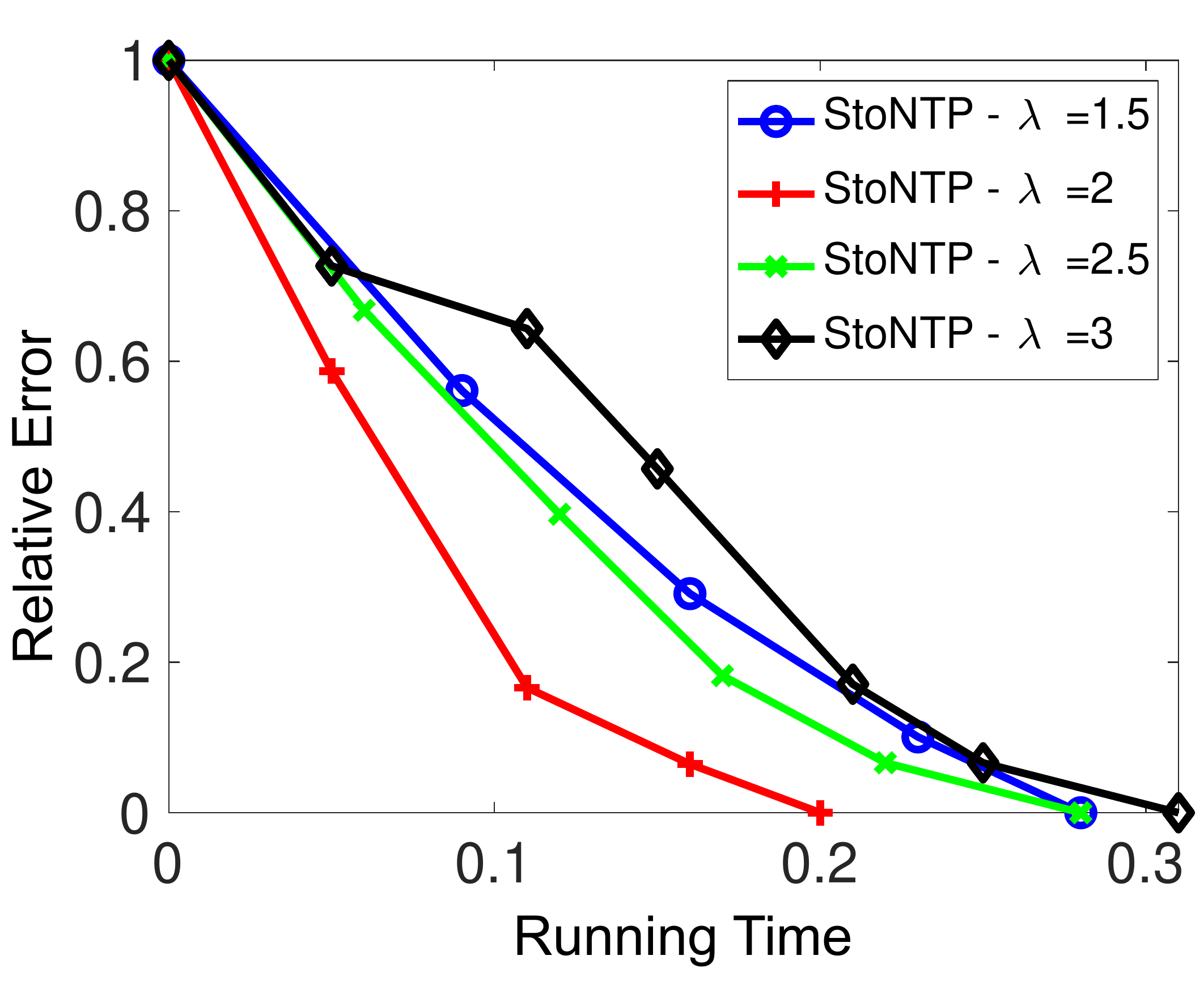}
\vspace{-8pt}
\caption{Test StoNTP with different $\lambda$'s: $m=100$, $n=800$, $k=10$, $\alpha = 1$. Batch size for StoNTP is 10. {The best step size is $\lambda = 2$.}}
\label{fig:StoNTP_lamb}
\end{figure}

\begin{figure}
\centering
\includegraphics[width = 0.45\textwidth]{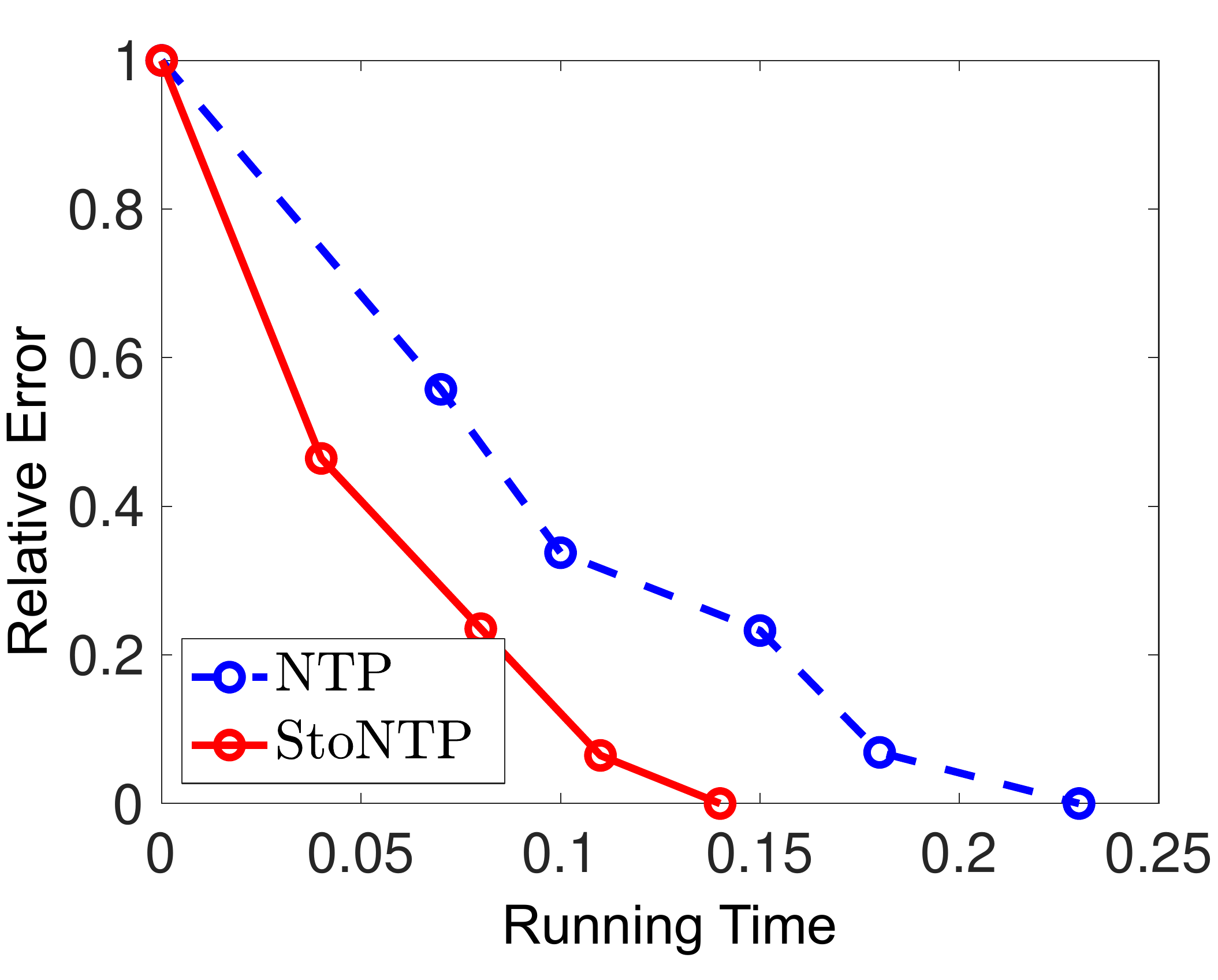}
\caption{NTP vs StoNTP: $m=100$, $n=800$, $k=10$. For NTP, we choose $\lambda = 2$, $\alpha = 5$. For StoNTP, we choose $\lambda = 2$, $\alpha = 1$. The batch size for StoNTP is 20.}
\label{fig:NTPvsStoNTP}
\end{figure}

\begin{figure}
\centering
\begin{tabular}{cc}
\includegraphics[width = 0.44\textwidth]{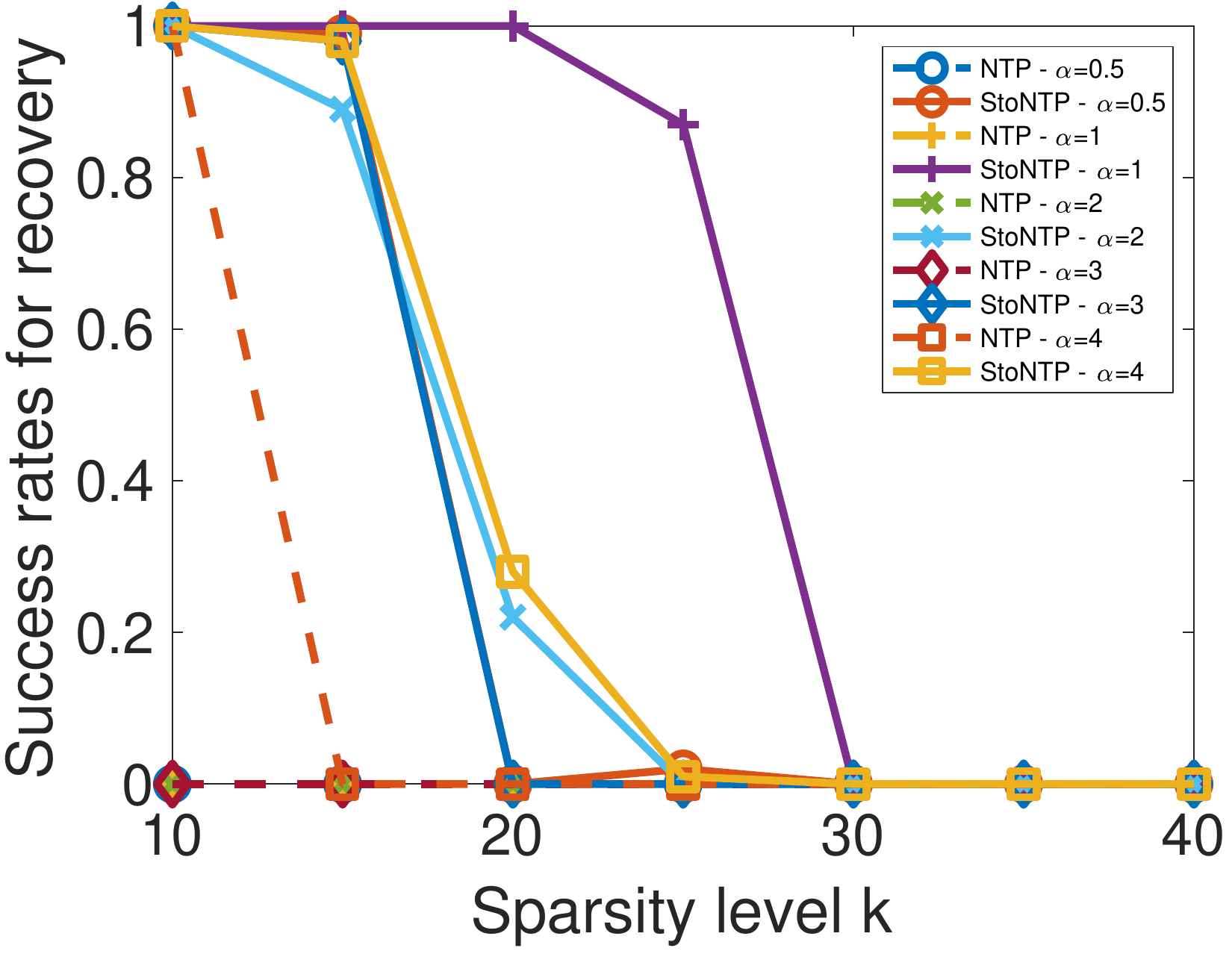}
&
\includegraphics[width = 0.43\textwidth]{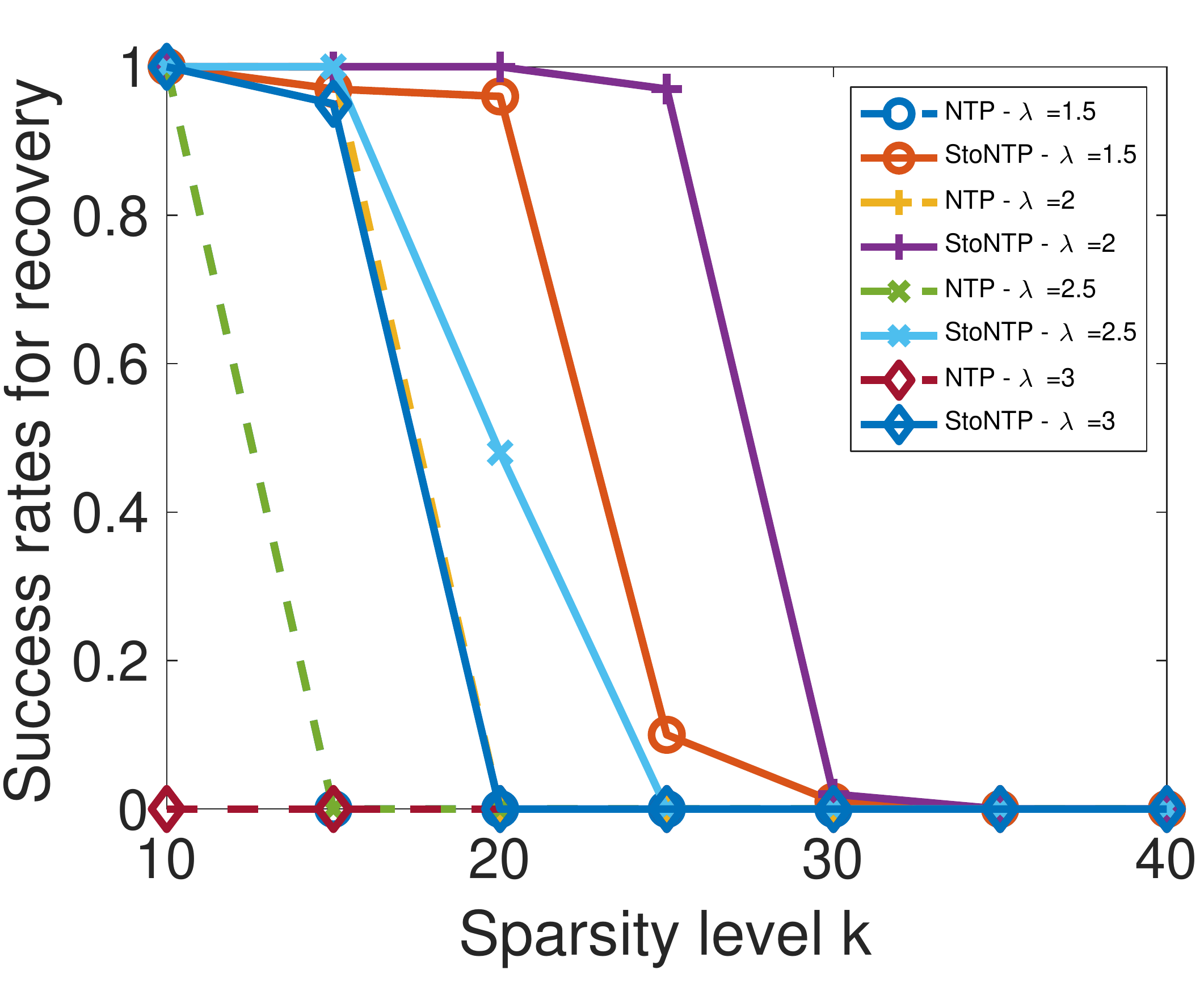}\\
StoNTP: $bs=30$, $\lambda=2$ &
NTP: $bs=30$, $\lambda=2$, StoNTP $\alpha=1$
\end{tabular}\\[10pt]
\begin{tabular}{c}
 \includegraphics[width = 0.44\textwidth]{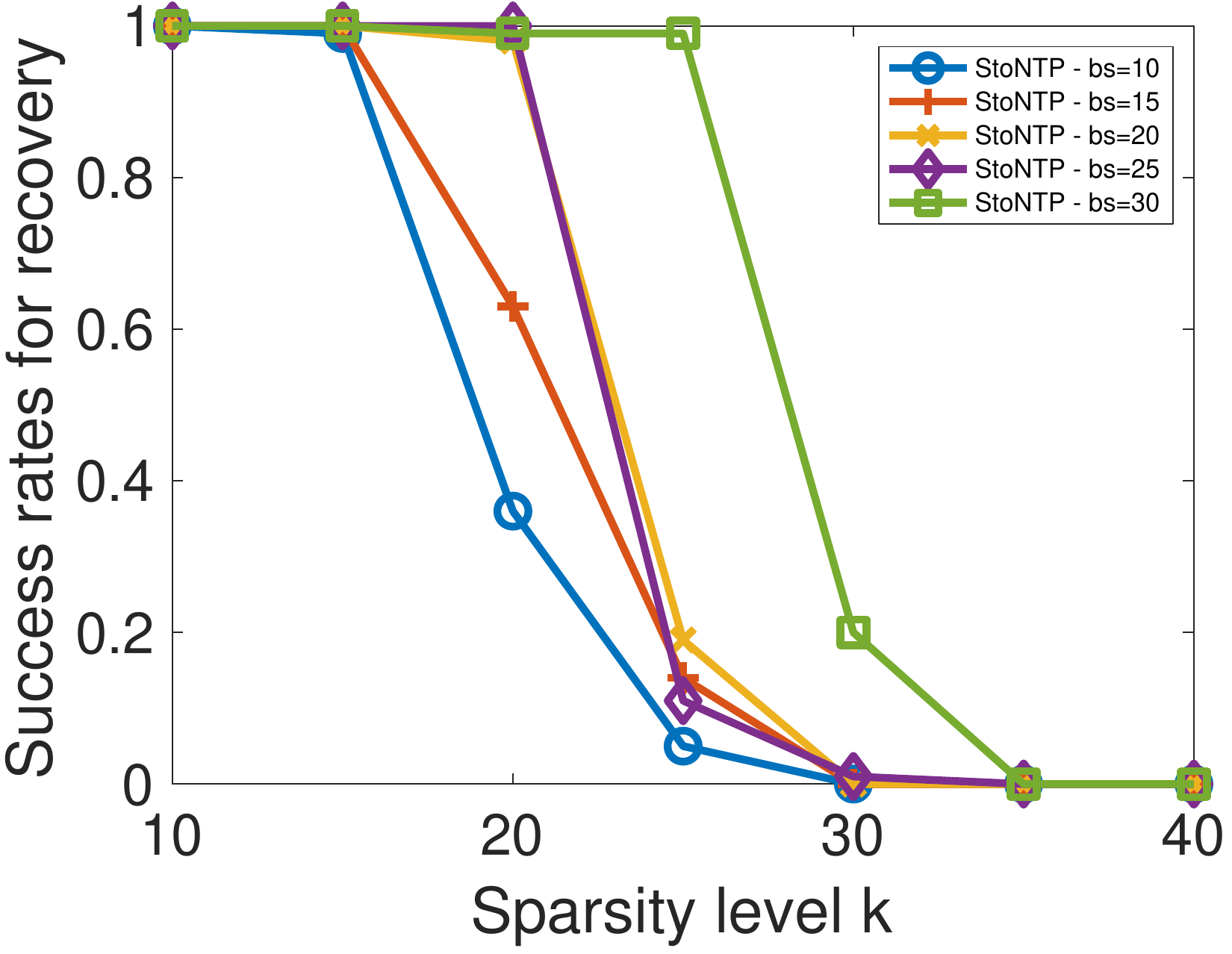}\\
 StoNTP: $\alpha=1$, $\lambda=2$.
\end{tabular}
\caption{Success rates for NTP and StoNTP with different parameters: $m=100$, $n=800$, $\alpha\in\{0.5,1,2,3,4\}$, $\lambda\in\{1.5,2,2.5,3\}$, batch size $bs\in\{10,15,20,25,30\}$. }
\label{fig:tune_para}
\end{figure}

\subsection{Nonlinear Measurements}
We extend the measurements from the linear case to the nonlinear case, and consider the $L_2$-regularized logistic regression model and support vector machine (SVM). In what follows, we set $a=5$,
$\epsilon = 10^{-3}$, and batch size to be 20. We also select each component function uniformly.

First, we consider the logistic regression model with the following objective function
$
f(\vx)=\frac1m\sum_{i=1}^m\log(1+\exp(-2y_i(\va_i\vx))),
$
where $\va_i$ represents the $i$-th row from the measurement matrix $A\in\R^{100\times 800}$ and classifiers $y_i\in\{-1,1\}$ such that $y_i=1$ with probability $p=\exp(\va_i\vx^*)/(1+\exp(\va_i\vx^*))$ for a fixed $\vx^*$ (the solution). The performance of our algorithms is shown in Fig.~\ref{fig:logreg}. In this experiment, the vectors $a_i$ are drawn i.i.d. from a Gaussian distribution and normalized to have unit norm. We set $m = 100$, $n = 800$, and $k=40$. For NTP the step size is $\lambda = 10$ and the step size for StoNTP is $\lambda = 30$. As shown in Fig.~\ref{fig:svm}, both StoNTP and NTP can attain a zero misclassification error. Notably, StoNTP can obtain a smaller loss than NTP.

\begin{figure}[H]
    \centering    \includegraphics[width=.45\textwidth]{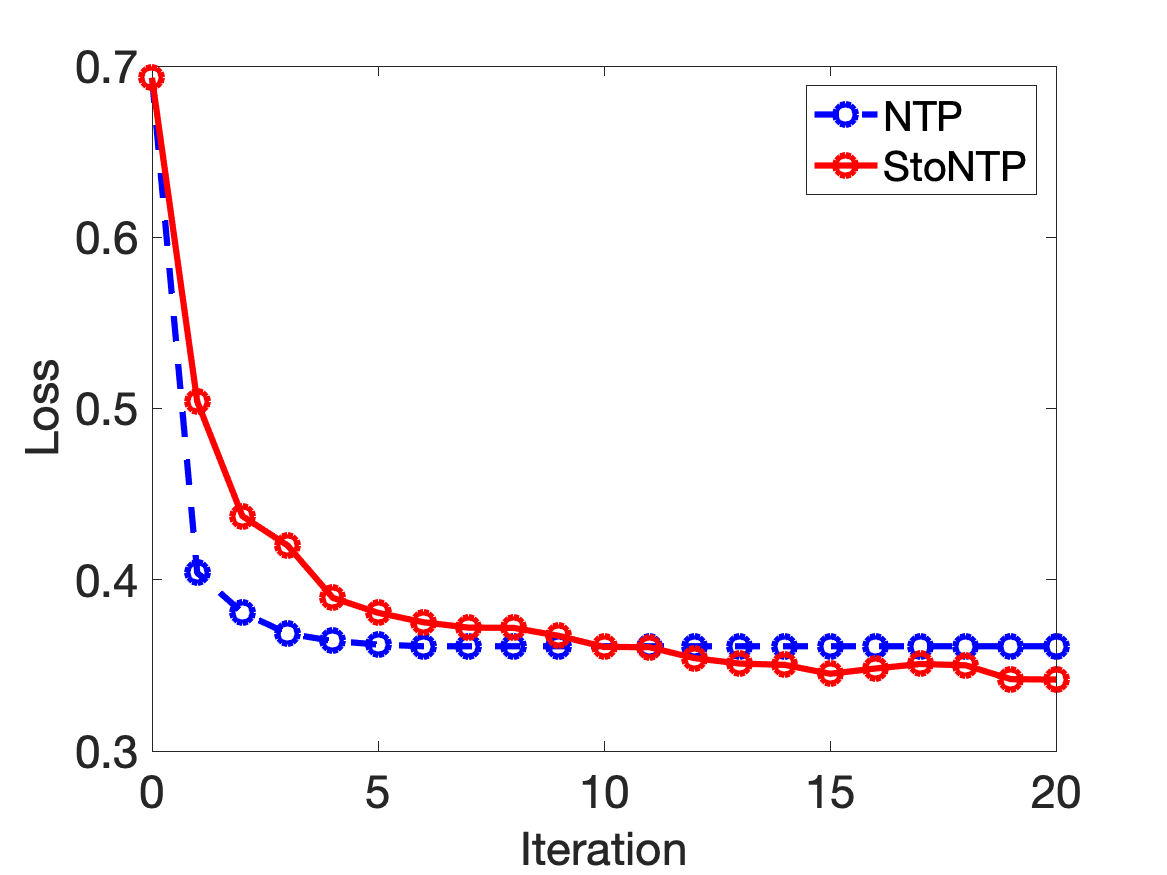}
    \includegraphics[width=.45\textwidth]{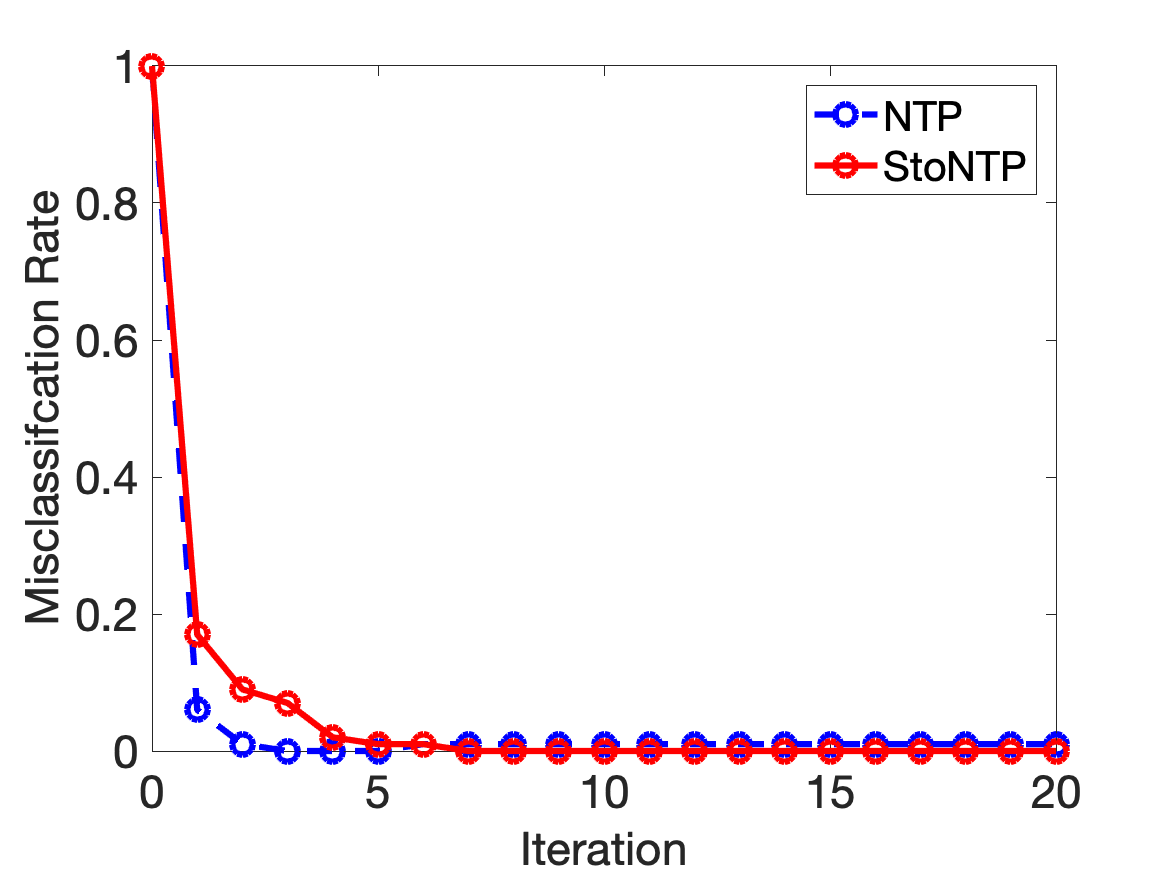}
    \vspace{-8pt}\caption{Comparison of NTP and StoNTP for logistic regression: convergence of loss (top) and misclassification rate (bottom).}
    \label{fig:logreg}
\end{figure}

Next, we consider the SVM problem with
$
f(\vx)=\frac1{2m}\sum_{i=1}^m(\max\{0,1-y_i\va_i\vx\})^2
$
where $\va_i$'s, $y_i$'s are defined as before. We set $m = 100$, $n = 800$, and $k=40$, and obtained the results in Fig.~\ref{fig:svm}. For NTP the step size is $\lambda = 10$, and the step size for StoNTP is $\lambda = 20$.  As shown in Fig.~\ref{fig:svm}, both StoNTP and NTP can attain a zero misclassification error. Notably, StoNTP can obtain a smaller loss than NTP. The vectors $a_i$ are drawn i.i.d. from a Gaussian distribution and normalized to have a unit norm.

\begin{figure}[H]
    \centering
    \includegraphics[width=.45\textwidth]{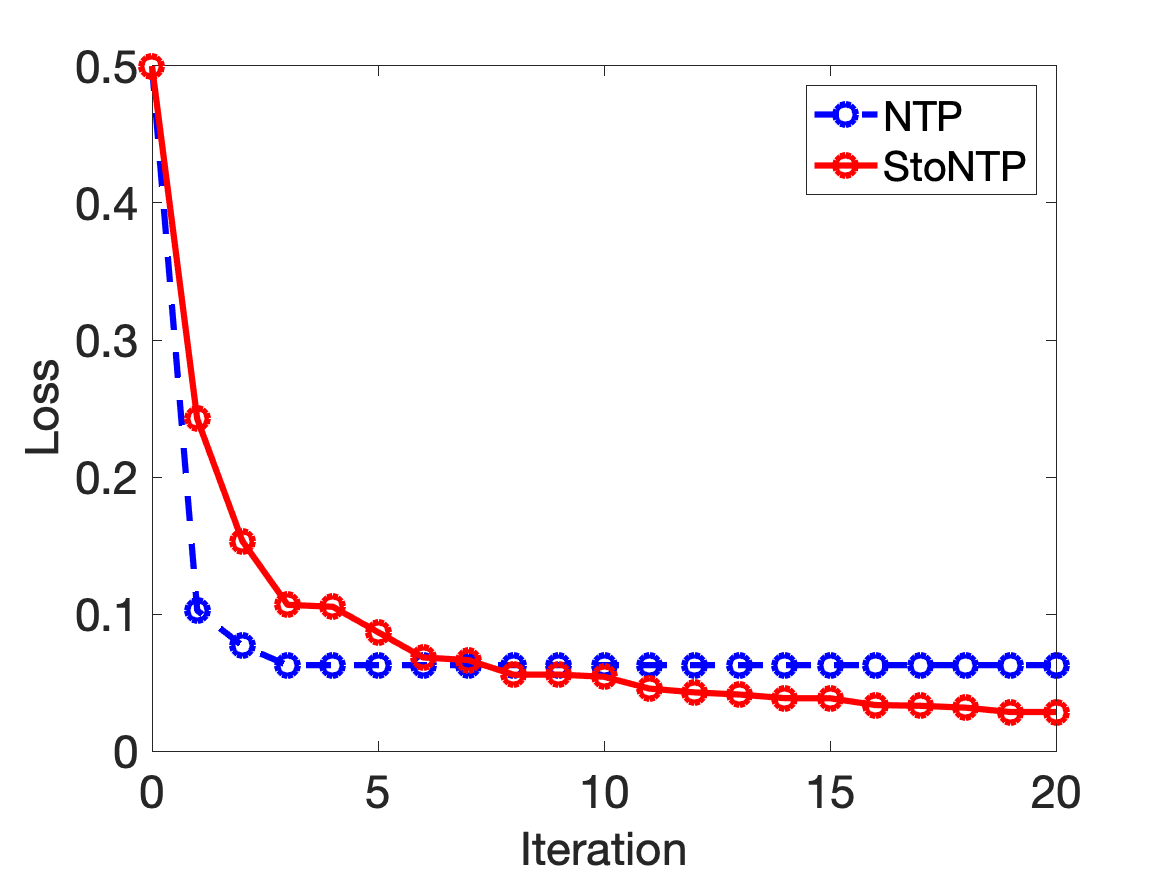}
    \includegraphics[width=.45\textwidth]{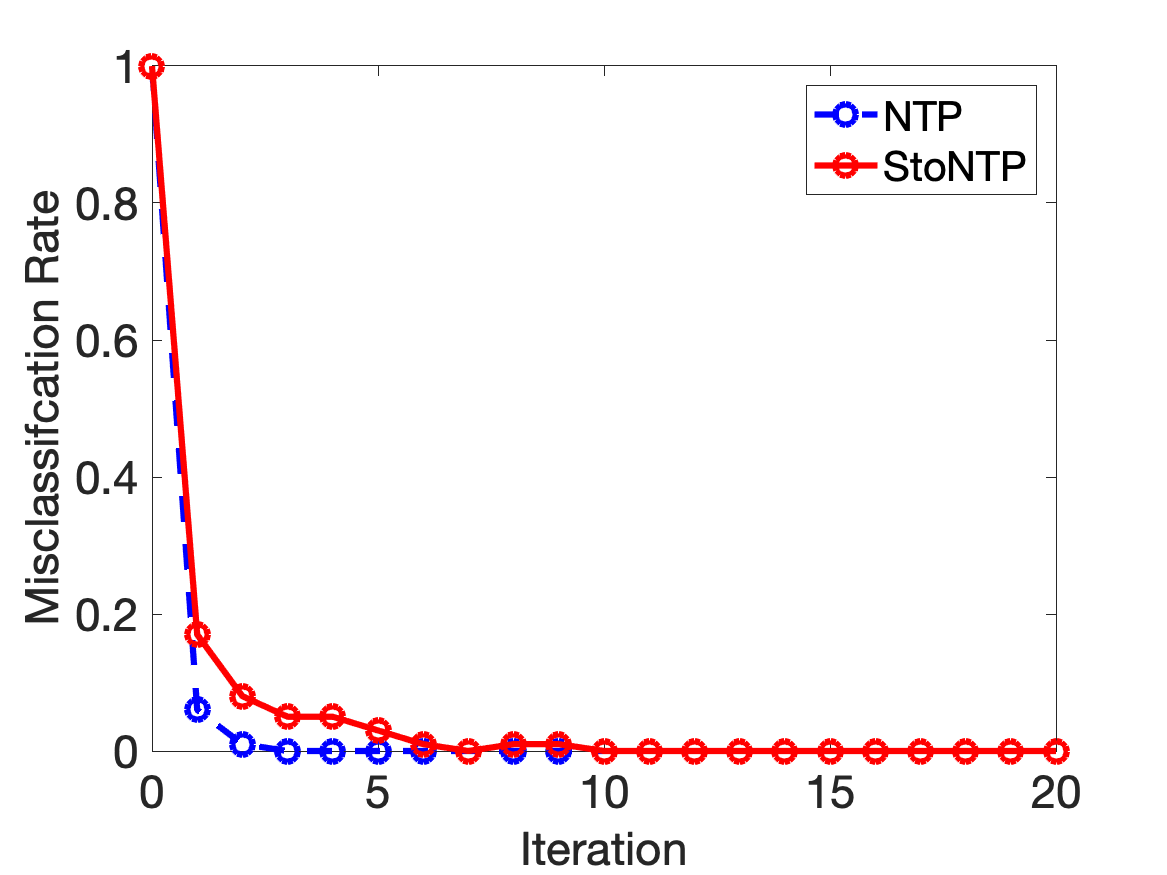}
    \vspace{-8pt}\caption{Comparison of NTP and StoNTP for SVM: convergence of loss (top) and misclassification rate (bottom).}
    \label{fig:svm}
\end{figure}

\section{Conclusion}
In this paper, we propose two stochastic natural thresholding  algorithms, i.e., StoNT and StoNTP, by extending the natural thresholding from the linear case to a general one and from the deterministic version to the stochastic one. Numerical simulations on linear and nonlinear measurements have shown the great potential of our algorithms in improving the recovery accuracy and computational efficiency.

\section*{Acknowledgements} DN was partially supported by NSF DMS 2011140 and NSF DMS 2108479. JQ was supported by NSF DMS 1941197.

\bibliographystyle{plain}
\bibliography{ref.bib}
\end{document}